\documentclass[letterpaper,10pt,conference]{ieeeconf}

\usepackage{amsmath,amssymb,amsfonts}
\usepackage{algorithm}
\usepackage[noend]{algpseudocode}
\usepackage{color}
\usepackage{tensor}
\usepackage{caption}
\usepackage{graphicx}
\usepackage{comment}
\usepackage{braket}
\newtheorem{theorem}{Theorem}

\newtheorem{lemma}{Lemma}

\DeclareMathOperator{\Tr}{Tr}

\usepackage[outdir=./]{epstopdf} 

\title{Sensitivity Bounds for Quantum Control and Time-Domain Performance Guarantees}
\author{S.\,P.\ O'Neil$^{1,*}$, \and
E.\,A.\ Jonckheere$^1$, \and S.\ Schirmer$^2$
\thanks{$^1$ Dept of Electrical \& Computer Engineering, University of Southern California, CA, USA. 
{\tt jonckhee@usc.edu, seanonei@usc.edu}}
\thanks{$^2$ Faculty of Science \& Engineering, Physics, Swansea University, UK. {\tt s.m.shermer@gmail.com}}
\thanks{This work was supported by Supercomputing Wales project "Robust Control Design for Quantum Technology".}
}

\begin{document}
\maketitle
\begin{abstract}
Control of quantum systems via time-varying external fields optimized to maximize a fidelity measure at a given time is a mainstay in modern quantum control. However, save for specific systems, current analysis techniques for such quantum controllers provide no analytical robustness guarantees. In this letter we provide analytical bounds on the differential sensitivity of the gate fidelity error to structured uncertainties for a closed quantum system controlled by piecewise-constant, optimal control fields.  We additionally determine those uncertainty structures that result in this worst-case maximal sensitivity.  We then use these differential sensitivity bounds to provide conditions that guarantee performance, quantified by the fidelity error, in the face of parameter uncertainty.  
\end{abstract}


\section{Introduction}
Control based on the optimal tuning of time-varying external fields forms the backbone of effective control for current quantum technology~\cite{Koch_2022}.  Such open-loop techniques have been successfully applied in areas from the implementation of quantum gates based on superconducting qubits~\cite{Werninghaus_2021}, to quantum metrology~\cite{Hou_2021}, and temperature and magnetic field sensing via large defect ensembles of Nitrogen-Vacancy Centers in diamond~\cite{Propson_2022} among many others.  Despite these successes of quantum optimal control, robustness issues remain.  Even with the best modeling, uncertainty in the underlying Hamiltonian persists~\cite{Koch_2022}, generating the potential for under-performance of an optimized controller as the plant diverges from the model system. 

Robust control \emph{design} techniques provide options to improve performance in the face of uncertainty. However, despite the success of robust control methods since the early 1980s, no comparable coherent theory for robust quantum control exists. One of the triumphs of classical robust control is the provision of analytic bounds on uncertainty that guarantee performance within specified bounds, best exemplified by the structured singular value $\mu$~\cite{robust_and_optimal_control}. With few exceptions, such as the class of optical systems amenable to $H_\infty$ methods in~\cite{wang_2023} or single-qubit gates with an uncertain interaction Hamiltonian as in~\cite{Kosut_2013}, analytic bounds to guarantee robust performance are not considered. Rather, the current state of the art in quantum robust control is largely a process of building robustness into synthesis via an optimization penalty based on the uncertainty model, followed by verification of robust performance bounds through post-design Monte Carlo simulation~\cite{Ge_2019,Koswara_2021,Koswara_2021_NJP,Daems_2013,Van_Damme_2017,RIM2022}. 

This touches on the need for a formal theory of robust quantum control. The survey~\cite{Petersen2019} highlights robustness issues with open-loop robust control of quantum systems and the nascent stage of development of robustness considerations in quantum control. In~\cite{Horowitz_1975}, Horowitz makes the point that a good theory of sensitivity should guarantee tolerances on performance over a range of parameter uncertainty, while in~\cite{Safonov_1981} Safonov outlines the performance and robustness trade-offs inherent in the application of frequency-domain robust control. This letter is a contribution in the spirit of such a theory for quantum control. We seek to provide reliable, analytic bounds on the differential sensitivity of the fidelity error to uncertainty in the drift and interaction Hamiltonians for fidelity-optimized controllers.  We then use these sensitivity bounds to provide guarantees on the gate fidelity in the face of such uncertainty. Further, since coherent quantum systems must remain oscillatory, not converging to a classical steady state and hence marginally stable, to retain their quantum properties, time-domain techniques are the natural choice to approach this analysis task at the edge of stability.  This is the approach taken here, departing from the established frequency-domain techniques of classical robust control~\cite{robust_and_optimal_control}.

The nominal system model and performance metrics are defined in Section~\ref{sec: prelims}, followed by the uncertainty model in Section~\ref{sec: uncertainty_model}.  The differential sensitivity of the controlled system for a given uncertainty model is derived in Section~\ref{sec: differential_sensitivity}, followed by derivation of bounds on the differential sensitivity and computation of the uncertainty structures that maximize the sensitivity in Section~\ref{sec: sensitivity_bounds}.  In Section~\ref{sec: performance} we leverage these bounds to provide performance guarantees in terms of the structure and size of the uncertainty.  The results are illustrated for a gate fidelity optimization problem in Section~\ref{sec: example}.

\section{Preliminaries}\label{sec: prelims}
Consider a closed quantum system of $Q$ qubits with underlying Hilbert space of dimension $N = 2^Q$. The state is characterized by the (pure) quantum state $\psi(t) \in \mathbb{C}^N$ with nominal drift Hamiltonian $H_0 \in \mathbb{C}^{N \times N}$. To control the evolution of $\psi(t)$, introduce $M$ control fields that optimally steer the trajectory from an initial state $\psi_0$ to a final state $\psi_f$ at a readout time $t_f$. The control fields take on constant values at $\kappa$ uniform time intervals of length $\Delta$ so that $t_k=k \Delta$ with $t_0=0$ and $t_f=\kappa \Delta$.  While the assumption of piecewise-constant controls may appear restrictive, such optimal pulse sequences are standard for optimal control problems~\cite{Koch_2022,KHANEJA_2005}.  Each control pulse $f_{m}^{(k)} \in \mathbb{R}$ enters the dynamics through an interaction Hamiltonian $H_m$ so that in the $k$th interval
\begin{equation}
    H^{(k)} = H_0 + \sum\limits_{m = 1}^M H_m f_{m}^{(k)},
    \quad 1 \leq k \leq \kappa.
\end{equation} 
Here $f_{m}^{(k)}$ is the strength of the $m$th control field in the time interval $t_{k-1}\leq t < t_k$. Let $u_k$ be a pulse of unit magnitude that is only non-zero for the interval $t_{k-1} \leq t \leq t_k$ (if $u(t)$ is the standard unit step then $u_k = u(t_{k-1}) - u(t_k)$). Then the total Hamiltonian is the sum of $H^{(k)}$ over time 
\begin{equation}
    H(t) = \sum_{k=1}^\kappa H^{(k)} u_k 
         = H_0 + \sum_{k=1}^\kappa\sum_{m=1}^M H_m f_{m}^{(k)} u_k. 
\end{equation}
The pure quantum state dynamics are governed by 
\begin{equation}
    \dot{\psi}(t) = -\frac{i}{\hbar} H(t)\psi(t) \ , \ \psi(0) = \psi_0
\end{equation}
with solution at the gate operation time $t_f$ 
\begin{align}
\begin{split}
    \psi(t_f) 
    &= \Phi^{(\kappa,0)}\psi_0 = \prod_{k=1}^\kappa \Phi^{(k,k-1)}\psi_0 = \\
    &\Phi^{(\kappa,\kappa-1)}\Phi^{(\kappa-1,\kappa-2)} \hdots \Phi^{(1,0)} \psi_0
\end{split}
\end{align}
where $\prod_{k=1}^{\kappa}$ indicates an ordered product.  Choosing units such that $\hbar=1$, the solution to the Schr\"odinger equation on the interval $[t_{k-1},t_k]$ is
\begin{equation}\label{eq: nominal_phi_matrix}
  \Phi^{(k,k-1)}  = \exp\left[-iH^{(k)}(t_k - t_{k-1})\right],
\end{equation}
and $\Phi^{(\kappa,0)}$ is the concatenation of the state transition matrices from $k=1$ to $\kappa$ covering the evolution from $t=0$ to $t = t_f$.

Equivalently, we can optimize for a unitary operator $U(t_f) = \Phi^{(\kappa,0)} \in \mathbb{U}(N)$ with target $U_f \in \mathbb{U}(N)$ where $\psi(t_f) = U(t_f)\psi_0$.  In this case the dynamics are governed by $\dot{U}(t) = -i H(t)U_0$ with $U_0 = I$, and $I$ is the identity on $\mathbb{C}^{N}$.
The solution is 
\begin{equation}\label{eq: schrodinger_gate}
    U(t_f) = \Phi^{(\kappa,0)} = \prod_{k=1}^{\kappa}\Phi^{(k,k-1)}. 
\end{equation}
The figure of merit is the normalized gate fidelity at $t_f$,
\begin{equation}
    \mathcal{F}(t_f) = \frac{1}{N}\left| \Tr \left[ U_f^{\dagger}\Phi^{(\kappa,0)} \right] \right|.
\end{equation}
The corresponding fidelity error is $e(t_f) = 1 - \mathcal{F}(t_f)$. 

\section{Uncertainty Model}\label{sec: uncertainty_model}

Consider an uncertain parameter in the nominal Hamiltonian modeled as $\delta \hat{H}_\mu$ where $\delta \in [\delta_1,\delta_2]$ is the scalar deviation of the uncertain parameter from its nominal value, structured as $\hat{H}_{\mu} \in \mathbb{C}^{N \times N}$ and normalized such that $\left\| \hat{H}_\mu \right\|_{F} = 1$. $\hat{H}_\mu$ is Hermitian as the Hamiltonian of a closed system is constrained to be Hermitian under the uncertainty. The uncertain Hamiltonian is then
\begin{subequations}\label{eq: perturbed_hamiltonian}
\begin{align}
    \tilde{H}(t) &= \sum_{k=1}^\kappa \tilde{H}^{(k)} u_k, \\ 
    \tilde{H}^{(k)}  &= H_0 + \sum_{m=1}^{M} H_m f_{m}^{(k)} + \delta \hat{H}_\mu\alpha_{\mu}^{(k)}.
    \end{align}
\end{subequations}
We consider the following uncertainty cases: 
\begin{itemize}
    \item Internal uncertainty in $H_0$: In this case $\tilde{H}_0=H_0 + \delta \hat{H}_0$ so that $\hat{H}_\mu = \hat{H}_0$, $\alpha_{\mu}^{(k)} = 1$ for all $k$, and $\hat{H}_0 = H_0/\left\| H_0 \right\|_{F}$ is the normalized structure matrix for $H_0$.  
    \item External uncertainty in interaction Hamiltonian $H_m$: In this case $f_{m}^{(k)} \tilde{H}_m = f_{m}^{(k)}(H_m + \delta \hat{H}_m)$ so that the final term in~\eqref{eq: perturbed_hamiltonian} is $\delta f_{m}^{(k)} \hat{H}_m$ and $\alpha_{\mu}^{(k)} = \alpha_{m}^{(k)} = f_{m}^{(k)}$, where
    $\hat{H}_\mu = \hat{H}_m$ is the normalized structure matrix for $H_m$. 
\end{itemize}

The perturbed solution to Eq~\eqref{eq: schrodinger_gate} at $t_f$ is given by 
\begin{subequations}
\begin{align} 
  \tilde{U}(t_f) &= \tilde{\Phi}^{(\kappa,0)} = \prod_{k=1}^{\kappa}\tilde{\Phi}^{(k,k-1)}, 
  \label{eq: perturbed_phi}\\
  \tilde{\Phi}^{(k,k-1)} &= \exp\left[-i\tilde{H}^{(k)} (t_k -t_{k-1}) \right].
  \label{eq: perturbed_Phi2}
\end{align}
\end{subequations}
The perturbed fidelity error due to the uncertainty $\hat{H}_\mu$ is 
\begin{equation}\label{eq: perturbed_error}
    \tilde{e}_{\mu}(t_f) 
    = 1 - \tilde{\mathcal{F}_\mu}(t_f) 
    = 1 - \frac{1}{N} \left| \Tr\left[U_f^\dagger \tilde{\Phi}^{(\kappa,0)} \right] \right|.
\end{equation}

\section{Differential Sensitivity}\label{sec: differential_sensitivity}
Taking the derivative of $\tilde{e}_\mu(t_f)$ with respect to the uncertain parameter $\delta$ structured as $\hat{H}_\mu$ yields the following from Eq.~(11) of~\cite{Floether_2012} and Eq.~(28) of~\cite{Schirmer_2011} when evaluated at $\delta = 0$
\begin{equation}\label{eq: sensitivity}
\left.\frac{\partial{\tilde{e}_{\mu}(t_f)}}{\partial \delta} \right|_{\delta=0} 
= - \frac{  \Re \left\{\Tr\left[U_f (\Phi^{(\kappa,0)})^\dagger\right]
\Tr\left[ U_f^\dagger \sum_{k=1}^{\kappa} \Lambda^{(\kappa,k)} \right] \right\}
}{ N \left| \Tr\left[U_f^\dagger \Phi^{(\kappa,0)} \right] \right|} . 
\end{equation}
Here $\Lambda^{(\kappa,k)}$ is defined as 
\begin{equation}
\Lambda^{(\kappa,k)} = \Phi^{(\kappa,k)} \frac{\partial \tilde{\Phi}^{(k,k-1)}}{\partial \delta}\Phi^{(k-1,0)},
\end{equation}
$\frac{\partial}{\partial \delta} \tilde{\Phi}^{(k,k-1)}$ is given by Eq.~(28) of~\cite{Schirmer_2011}
\begin{multline}\label{eq: dphi_tilde}
  \int_{t_{k-1}}^{t_{k}} e^{-iH^{(k)}(t_k-\tau)} \left(-i \hat{H}_\mu\alpha_{\mu}^{(k)}  \right) e^{-iH^{(k)} (\tau-t_{k-1})}d \tau 
\end{multline}
and $\alpha_{\mu}^{(k)} \in \set{ 1,f_{m}^{(k)} }$ based on the type of uncertainty as detailed in Section~\ref{sec: uncertainty_model}.  We equivalently write Eq.~\eqref{eq: sensitivity} as
\begin{equation}\label{eq: reduced_sensitivity_matrix}
 \left.\frac{\partial{\tilde{e}_{\mu}(t_f)}}{\partial \delta} \right|_{\delta = 0}  
 = \Re \left\{ -\frac{e^{-i\phi}}{N}  \Tr\left[ U_f^\dagger \sum_{k=1}^{\kappa} \Lambda^{(\kappa,k)} \right] \right\},
\end{equation}
where $\phi = \angle \Tr\left[ U_f^{\dagger} \Phi^{(\kappa,0)} \right]$.  For brevity in what follows, we define $\frac{\partial}{\partial\delta}\tilde{e}_{\mu}(t_f) = \zeta_{\mu}(t_f)$ as the derivative of fidelity error in the direction $\hat{H}_{\mu}$ evaluated at $\delta = 0$. 

\section{Differential Sensitivity Bounds}\label{sec: sensitivity_bounds}
We now consider bounds on the size of the differential sensitivity. To provide an initial bound on the differential sensitivity, we directly bound the absolute value of Eq.~\eqref{eq: reduced_sensitivity_matrix}:
\begin{align}
    \left| \zeta_{\mu}(t_f) \right| 
    &\leq \left| \frac{e^{-i \phi}}{N} \right| \cdot \left| \Tr \left[ U_f^{\dagger} \sum_{k = 1}^{\kappa} \Lambda^{(\kappa,k)} \right] \right| \nonumber\\ 
    &\leq \frac{1}{N} \sum_{k=1}^{\kappa} \left| \Tr \left[ U_f^{\dagger} \Lambda^{(\kappa,k)} \right] \right|. 
\end{align} 
Invoking the von Neumann trace inequality~\cite{Mirsky_1975}, and noting that all singular values $\sigma_\ell(U_f^\dagger)$ of $U_f^\dagger$ are $1$, gives the bound 
\begin{equation}\label{eq: bound_B_1}
    \left| \zeta_{\mu}(t_f) \right| \leq \frac{1}{N} \sum_{k=1}^{\kappa} \sum_{\ell=1}^{N} \sigma_\ell \left( \Lambda^{(\kappa,k)} \right).
\end{equation} 
As all $\Phi^{(\kappa,k)}$ are unitary have
\begin{multline}
{\sigma}_{\ell}(\Lambda^{(\kappa,k)}) 
= \sigma_\ell \left( \Phi^{(\kappa,k)} \frac{\partial \tilde{\Phi}^{(k,k-1)}}{\partial \delta} \Phi^{(k-1,0)} \right)
\\ 
= \left\| \frac{\partial \tilde{\Phi}^{(k,k-1)}}{\partial \delta} \right\|_2 \leq \int_{t_{k-1}}^{t_k} \!\!\left\| e^{-iH^{(k)}(t_k-\tau)}  \right\|_2 
\times \left\| \alpha_{\mu}^{(k)} \hat{H}_\mu \right\|_2 \\
\times \left\| e^{-iH^{(k)} (\tau-t_{k-1})} \right\|_2 d \tau 
= \Delta \left| \alpha_{\mu}^{(k)} \right| \left\| \hat{H}_\mu \right\|_2,
\end{multline}
which yields the bound
\begin{equation}\label{eq: bound_B_2}
\left\| \zeta_{\mu}(t_f) \right\| 
\leq \Delta \left\| \hat{H}_\mu \right\|_2  \; \sum_{k=1}^{\kappa}\left| \alpha_{\mu}^{(k)} \right| =: B_1.
\end{equation}
This bound is conservative and only relevant to a specific perturbation structure $\hat{H}_\mu$. However, relying only on the system parameters, the bound is not constrained to a neighborhood of $\delta = 0$ and is independent of the size of the uncertainty $\delta$, making it applicable to worst-case performance analysis.  

To obtain tighter bounds in the perturbative regime about $\delta = 0$ that are applicable to more general uncertainty structures, we unpack the structure in~\eqref{eq: reduced_sensitivity_matrix}.  Employing the cyclic property and linearity of the trace, rewrite the last term of~\eqref{eq: reduced_sensitivity_matrix} as 
 \begin{align*}
     \Tr & \left[ U_f^{\dagger}\sum_{k=1}^\kappa \Lambda^{(\kappa,k)} \right] 
    = \sum_{k=1}^{\kappa} \Tr\left[\Phi^{(k-1,0)}U_f^{\dagger}\Phi^{(\kappa,k)} \frac{\partial \tilde{\Phi}^{(k, k-1)}}{\partial \delta} \right]
 \end{align*}
where $\hat{H}_\mu$ is restricted to a subset of the $N\times N$ Hermitian matrices, which is justified as the $\hat{H}_\mu$ are necessarily Hermitian and constrained to the drift and interaction Hamiltonian matrix structures by the assumptions of Section~\ref{sec: uncertainty_model}.  Define a basis for this subset as $\{\hat{H}_m \}_{m = 0}^{M}$ for $M < N^2$.  An arbitrary, normalized uncertainty structure is represented as 
\begin{equation}
\hat{H}_\mu = \sum_{m=0}^{M} s_m \hat{H}_m.
\end{equation}
If $\mathbf{s_\mu} \in \mathbb{R}^{M+1}$ is a column vector of the scalars $s_m$ and $\left\| \hat{H}_m \right\|_{F} = 1$ for each $m$, then retaining normalization requires $\left\| \hat{H}_\mu \right\|_{F} = \sqrt{\Tr\left[\hat{H}_\mu^{\dagger} \hat{H}_\mu \right]} = 1$. This holds if $\left\| \mathbf{s_\mu} \right\|_{2} =1$.  Substituting this expansion for $\hat{H}_{\mu}$ in~\eqref{eq: dphi_tilde} yields 
\begin{equation}\label{eq: dphi_tilde2}
  \frac{\partial \tilde{\Phi}^{(k,k-1)}}{\partial \delta} 
  = \sum_{m = 0}^M \mathbf{X}^{(k)}_m s_m 
\end{equation}
where $\mathbf{X}^{(k)}_m$ is given by
\begin{equation}\label{eq: X_ell}
  -i
  \int_{t_{k-1}}^{t_{k}} \!\! e^{-iH^{(k)}(t_k-\tau)} \hat{H}_m \alpha_{m}^{(k)} e^{-iH^{(k)} (\tau-t_{k-1})}d \tau. 
\end{equation}
Defining 
\begin{equation}\label{eq: Z_kl}
\Re \left\{ -\frac{e^{-i\phi}}{N} \Tr\left[\Phi^{(k-1,0)}U_f^\dagger \Phi^{(\kappa,k)} \mathbf{X}^{(k)}_m \right] \right\} =: Z^{(k)}_m 
\end{equation}
the differential sensitivity has the compact expression
\begin{equation}\label{eq: Z_Gamma}
\zeta_\mu(t_f) = \mathbf{1}^{T} \mathbf{Z}(t_f,\mathbf{f}) \mathbf{s_\mu} = \mathbf{\Gamma}(t_f,\mathbf{f})\mathbf{s_\mu},
\end{equation}
where $\mathbf{1} \in \mathbb{R}^{\kappa \times 1}$ is the vector of $\kappa$ $1$s that sum the vectorized components of $\frac{\partial}{\partial \delta} \tilde{\Phi}^{(k,k-1)}.$ The $t_f$ and $\mathbf{f}$ in $\mathbf{Z}$ and $\mathbf{\Gamma}$ indicate the dependence of these matrices on the gate operation time and control fields. 

Before deriving the improved upper bounds on $\zeta_\mu(t_f)$, note that \eqref{eq: Z_Gamma} facilitates two interpretations of the differential sensitivity based on the constraints of the uncertainty. If the uncertainty structure is assumed constant so that the Hamiltonian does not vary with time, $\mathbf{s}_\mu$ is a constant vector so that $\mathbf{\Gamma}(t_f,\mathbf{f})(\cdot)$ can be viewed as a function that accepts as input a single uncertainty structure, $\mathbf{s_\mu}$ and provides as output, the sensitivity in that direction. In this case, we provide the following upper bound on the differential sensitivity at $\delta = 0$:

\begin{theorem}\label{thm: theorem_1}
The maximum of $\left| \zeta_{\mu}(t_f) \right|$ in~\eqref{eq: Z_Gamma} with constant $\mathbf{s}_\mu$ is 
$B_2:=\left\| \mathbf{\Gamma} \right\|_{2}$. Further, the uncertainty that maximizes $\left| \zeta_{\mu}(t_f) \right|$ is given by $\bar{H}_\mu = \sum_{m = 0}^M v_m \hat{H}_m $ where $\{v_m\}$ are the components of the normalized vector $\hat{\mathbf{v}} = \mathbf{\Gamma}^T / B_2$. 
\end{theorem}

\begin{proof}
Following directly from~\eqref{eq: Z_Gamma} and for a normalized uncertainty structure $\left\| \mathbf{s_\mu} \right\|_2 = 1$, $\left| \zeta_{\mu}(t_f) \right| = \left\| \mathbf{\Gamma} \mathbf{s_\mu} \right\|_{2} \leq \left\| \mathbf{\Gamma} \right\|_{2} \cdot \left\| \mathbf{s_\mu} \right\|_2 = \left\| \mathbf{\Gamma} \right\|_2 = B_2$. Noting that $\mathbf{\Gamma}$ is simply a real $(M+1)$-dimensional row vector, the $\mathbf{s_\mu}$ that maximizes the inner product $\mathbf{\Gamma}\mathbf{s}_\mu$ with $\left\| \mathbf{s}_\mu \right\|_2 = 1$ is $\hat{\mathbf{v}} = \mathbf{\Gamma}^{T}/B_2$. The maximum uncertainty direction $\bar{H}_\mu$ in terms of the Hamiltonian uncertainty basis $\{\hat{H}_m\}$ directly follows.
\end{proof}

Alternatively, consider an uncertainty that is constant over each time step but varies over the evolution. Then for each time step $k$, $\hat{H}_\mu^{(k)} = \sum_{m=0}^{M} s_m^{(k)}\hat{H}_m$ with a corresponding $\mathbf{s}_\mu^{(k)}$ in the vectorized representation. The differential sensitivity is then
\begin{multline}
    \label{eq: Z_time_change}
    \zeta_{\set{\mu}}(t_f) 
    = \left.\frac{\partial \tilde{e}_{\set{\mu}}(t_f)}{\partial \delta} \right|_{\delta = 0} \\
    = \sum_{k=1}^{\kappa} \left (\sum_{m=0}^{M} Z^{(k)}_{m} s_m^{(k)} \right) 
    = \sum_{k=1}^{\kappa} \mathbf{Z}^{(k)} \mathbf{s}_{\mu}^{(k)} 
    = \sum_{k=1}^{\kappa} \varsigma^{(k)},
\end{multline}
where $\mathbf{Z}^{(k)}$ is the $k$th row of $\mathbf{Z}(t_f,\mathbf{f})$ in~\eqref{eq: Z_Gamma} and $\tilde{e}_{\set{\mu}}$ indicates that $\hat{H}_\mu$ is not fixed but given by the sequence $\set{\hat{H}_\mu^{(k)}}_{k=1}^{\kappa}$. This leads to the following alternative bound for the differential sensitivity. 

\begin{theorem}\label{thm: theorem_2}
    The maximum size of $\left| \zeta_{\set{\mu}}(t_f) \right|$ for the formulation of~\eqref{eq: Z_time_change} with uncertainty defined by the sequence $\set{\mathbf{s}_\mu^{(k)}}$ and $\left\| \mathbf{s}_{\mu}^{(k)} \right\| = 1$ for all $k$ is $\left\| \set{\bar{\varsigma}^{(k)}} \right\|_{\ell^1} =: B_3$ where $\bar{\varsigma}^{(k)} = \left\| \mathbf{Z}^{(k)} \right\|_{2}$. Further the sequence $\set{\bar{\mathbf{s}}_{\mu}^{(k)}}$ that achieves the bound $B_3$ is given by $\bar{\mathbf{s}}_{\mu}^{(k)} = {\mathbf{Z}^{(k)}}^T/\bar{\varsigma}^{(k)}$ for $1 \leq k \leq \kappa$. 
\end{theorem}
\begin{proof}
    Suppose $\mathbf{s}_{\mu}^{(k)}$ is normalized for all $k$ with differential sensitivity given by~\eqref{eq: Z_time_change}. Then $\left| \zeta_{\set{\mu}}(t_f) \right| = \left| \sum_{k=1}^{\kappa} \mathbf{Z}^{(k)} \mathbf{s}_{\mu}^{(k)} \right| \leq \sum_{k=1}^{\kappa} \left| \mathbf{Z}^{(k)} \mathbf{s}_{\mu}^{(k)} \right| = \left\| \set{ \varsigma^{(k)} } \right\|_{\ell^1}$. It follows that this sum is maximized if each term $\varsigma^{(k)}$ is of the same sign and takes on the maximum $\bar{\varsigma}^{(k)} = \lvert \max_{\mathbf{s}_{\mu}^{(k)}} \mathbf{Z}^{(k)} \mathbf{s}_{\mu}^{(k)} \rvert$ for each $k \in \set{1,2,\hdots,\kappa}$. 
    Since $\mathbf{s}_{\mu}^{(k)}$ is normalized, we have that $\bar{\varsigma}^{(k)} = \left\| \mathbf{Z}^{(k)} \right\|_{2}$. As such the maximum sensitivity is given by $B_3 := \sum_{k=1}^{\kappa} \left| \bar{\varsigma}^{(k)} \right| = \left\| \set{\bar{\varsigma}^{(k)}} \right\|_{\ell^{1}}$. For the sequence of uncertainty structures $\mathbf{s}_{\mu}^{(k)}$, we have that for each $k$, $\mathbf{Z}^{(k)} \mathbf{s}_{\mu}^{(k)}$ is maximized by $\mathbf{s}_\mu^{(k)} = {\mathbf{Z}^{(k)}}^T/\left\| \mathbf{Z}^{(k)} \right\| = {\mathbf{Z}^{(k)}}^T / \bar{\varsigma}^{(k)}$. The sequence $\set{\bar{\mathbf{s}}_{\mu}^{(k)}} = \set{{\mathbf{Z}^{(k)}}^T/\bar{\varsigma^{(k)}}}$ directly follows. 
\end{proof}

\section{Guaranteed Performance}\label{sec: performance}
We now seek to leverage the sensitivity bound $B_1$ of~\eqref{eq: bound_B_2} to guarantee performance in the face of uncertainty structured as $\hat{H}_\mu$ with strength $\delta$.  For a given controller, characterized by the set of control fields $\{f_{m}^{(k)}\}$ and fixed gate operation time $t_f$, we view the perturbed error $\tilde{e}_\mu (\delta)$ in the direction $\hat{H}_\mu$ as a function of the uncertain parameter $\delta$ where $t_f$ is suppressed in the expression.  We likewise consider the differential sensitivity as a function of $\delta$ so that 
\begin{equation*}
   \zeta_{\mu}(\delta) = \frac{\partial \tilde{e}_\mu (\delta)}{\partial \delta} = \lim_{\Delta \delta \rightarrow 0} \frac{\tilde{e}_\mu(\delta + \Delta \delta) - \tilde{e}_\mu(\delta)}{\Delta \delta}. 
\end{equation*}
Our goal is to determine a bound on $\delta \in [\delta_1, \delta_2]$ such that $\tilde{e}_\mu(\delta) \leq \epsilon$ for the error threshold $\epsilon$. Here $\delta_1$ and $\delta_2$ determine the endpoints of the uncertainty set for the parameter $\delta$ admitted by the physical model.  Before providing the main result, we establish a pair of lemmas. 

\begin{lemma}\label{lemmma_1}
On any interval $[\delta_1,\delta_2] \subset \mathbb{R}$ such that the fidelity $\tilde{\mathcal{F}_{\mu}}(\delta) \neq 0$,
the function $\tilde{e}_\mu(\delta)$ is locally Lipschitz.
\end{lemma}
\begin{proof}
To establish that $\tilde{e}_{\mu}(\delta)$ is locally Lipschitz it suffices to show that $\tilde{e}_\mu (\delta)$ is real analytic in $\delta$. We rewrite $-i \tilde{H}^{(k)}$ as $-i \left( A^{(k)} + \delta B^{(k)} \right)$ with
\begin{align*}
A^{(k)} = H_0 + \sum\limits_{m = 1}^{M}H_{m}f_{m}^{(k)}, \quad 
B^{(k)} = \hat{H}_\mu \alpha_{\mu}^{(k)}.
\end{align*}
Allowing a non-vanishing, complex perturbation $\eta = x+iy$ and considering deviations $\Delta x$ and $\Delta y$ we have $F(\eta):=\exp\left[-i(t_k-t_{k-1}) (A^{(k)} + \eta B^{(k)})\right]$ is complex analytic as 
\begin{multline*}
\left.\frac{\partial F(\eta)}{\partial x}\right|_{\eta \neq 0} 
= \left.\frac{\partial F(\eta)}{i\partial y}\right|_{\eta \neq 0} 
 = \\ -i\int_{t_{k-1}}^{t_{k}} \!\! e^{-i(t_k-\tau)(A^{(k)} + \eta B^{(k)})} B^{(k)} e^{-i (\tau-t_{k-1})(A^{(k)} + \eta B^{(k)})}d \tau.
\end{multline*}
Then $F(\eta)$ restricted to $\delta=\Re\{\eta\}$ is real analytic and has a convergent power series in $\delta \in [\delta_1,\delta_2]$. 
Given the product of real analytic functions is real analytic and, by the Fa\`a di Bruno formula, the composition of real analytic functions is real analytic~\cite{krantz2002primer}, $\Tr\left[ U_f^{\dagger} \tilde{\Phi}^{(\kappa,0)}(\delta) \right] =: g(\delta) \in \mathbb{C}$ is real analytic. Employing the same argument for $|g(\delta)| = \sqrt{g(\delta)g^*(\delta)}$, we have that $\left|\Tr\left[ U_f^{\dagger} \tilde{\Phi}^{(\kappa,0)}(\delta) \right]\right|$ is real analytic except for when $g(\delta) = 0$, where $\tilde{\mathcal{F}}_\mu(\delta) = 0$. It then follows that $\tilde{e}_\mu(\delta) = 1-\tilde{\mathcal{F}}_\mu(\delta)$ is locally Lipschitz on an arbitrary interval $[\delta_1,\delta_2]$ if $\tilde{\mathcal{F}}_\mu(\delta) \neq 0$.  
\end{proof}

\begin{lemma}\label{lemma_2}
On the interval $[\delta_1,\delta_2]$ and given that $\tilde{H}_k$ remains Hermitian, $\left| \zeta_{\mu}(\delta) \right|$ is uniformly bounded with Lipschitz constant $B_1$ established by~\eqref{eq: bound_B_2}. 
\end{lemma}

\begin{proof}
Given that $B_1$ is an upper bound at $\delta = 0$, it suffices to show that if $B_1$ is independent of $\delta$, the bound holds on an arbitrary interval $[ \delta_1, \delta_2 ]$. To establish the independence, consider evaluation of $\zeta_{\mu}(\delta)$ at some $\delta_0 \neq 0$.  Then we have from~\eqref{eq: dphi_tilde}
\begin{multline}\label{eq: perturbed_dphi}
\left.\frac{ \partial \tilde{\Phi}^{(k,k-1)}}{\partial \delta} \right|_{\delta = \delta_0} 
= \int_{t_{k-1}}^{t_{k}}  \!\! e^{-i\tilde{H}^{(k)}(\delta_0)(t_k-\tau)} \times \\ 
\left(-i \alpha_{\mu}^{(k)} \hat{H}_\mu \right) e^{-i \tilde{H}^{(k)}(\delta_0) (\tau-t_{k-1})}d \tau 
\end{multline}
where the perturbed Hamiltonians are
\begin{equation}\label{eq: perturbed_H_k}
\tilde{H}^{(k)}(\delta_0)  
    = H_0 + \sum_{m=1}^M H_m f_{m}^{(k)} + \alpha_{\mu}^{(k)} \hat{H}_{\mu}\delta_0.
\end{equation}
Given that $\hat{H}_\mu$ is still Hermitian, it follows that the terms $ e^{-i\tilde{H}^{(k)}(\delta_0)(t_k-\tau)}$ and $e^{-i\tilde{H}^{(k)}(\delta_0) (\tau-t_{k-1})}$ in~\eqref{eq: perturbed_dphi} are still unitary, and the bound in~\eqref{eq: bound_B_2} remains unchanged. As such we have $\left| \zeta_{\mu}(\delta) \right| \leq B_1$ on $[\delta_1,\delta_2]$.   
\end{proof}

We now state the main result and provide a bound on $\delta$ to guarantee a given performance requirement. 
\begin{theorem}\label{thm: theorem_3}
Given a maximum allowable error $\epsilon$ and an uncertainty structure $\hat{H}_\mu$, the perturbed fidelity error $\tilde{e}_\mu(\delta)$ is less than $\epsilon$ for all $|\delta|<\bar{\delta}=\frac{\epsilon - e(0)}{B_1}$.  
\end{theorem}
\begin{proof}
From Lemma~\ref{lemmma_1}, $\tilde{e}_\mu(\delta)$ is locally Lipschitz on $[\delta_1,\delta_2]$. It follows that $\left| \tilde{e}_\mu(\delta_a) - \tilde{e}_{\mu}(\delta_b) \right| \leq L_{ab} \left| \delta_a - \delta_b \right|$ on each open interval $(\delta_a,\delta_b) \subset [\delta_1,\delta_2]$ with Lipschitz constant $L_{ab}$. From Lemma~\ref{lemma_2}, the bound $B_1$ for a given $\hat{H}_\mu$ for a given controller provides a uniform upper bound on $L_{ab}$ over the interval $[\delta_1,\delta_2]$ so that $\tilde{e}_\mu(\delta)$ is Lipschitz on $[\delta_1,\delta_2]$ with Lipschitz constant $B_1$. Then, since $\delta = 0$ (the ``nominal" uncertainty) necessarily lies in the interval $[\delta_1,\delta_2]$ of allowable uncertainty, $\left| \tilde{e}_\mu(\delta) - e(0) \right| \leq B_1 \left| \delta - 0\right|$. Since the fidelity error is always non-negative, the perturbed error is bounded by $\tilde{e}_\mu(\delta) \leq e(0) + B_1 |\delta|$ where $e(0)$ is the nominal error. We now have the performance condition $\tilde{e}_\mu(\delta) \leq e(0) + B_1 |\delta| \leq \epsilon $ from which the bound $| \delta | \leq \frac{\epsilon - e(0)}{B_1}$ guarantees $ \tilde{e}_\mu(\delta) \leq \epsilon$.
\end{proof}

Although the bound $\bar{\delta}=\frac{\epsilon - e(0)}{B_1}$ guarantees that $\tilde{e}_\mu (\delta)$ does not exceed the threshold $\epsilon$, this yields a highly conservative performance bound on $\delta$. However, we can apply an iterative process based on the local bounds of Theorem~\ref{thm: theorem_2} to compute a less conservative worst-case perturbation. Quantizing the uncertainty size $\delta$ into uniform steps of a given magnitude $\mathbf{d}$, we compute the directions of maximum sensitivity of the fidelity error over every time interval $k$ at increasing perturbation strength $n\mathbf{d} = \delta_n$.  These maximum sensitivity directions from Theorem~\ref{thm: theorem_2} take the form $\mathbf{s}_{\mu}^{(k)}(n) = \left[s_0^{(k)}(n),s_1^{(k)}(n),\hdots,s_{M}^{(k)}(n)\right]^T$ generated by the Hamiltonian $\tilde{H}^{(k)}(\delta_n)$ at perturbation strength $\delta_n$. Initializing the procedure with the worst-case perturbed Hamiltonian resulting from the sequence $\set{s_\mu^{(k)}(0)} = \set{\mathbf{\bar{s}_\mu}^{(k)}}$ of Theorem~\ref{thm: theorem_2} we have: 
\begin{algorithm}
\begin{algorithmic}[1]\label{algorithm}
\caption{Compute largest $\bar{\delta}$ such that $\tilde{e}_\mu(\bar{\delta}) < \epsilon$}
\State{$n = 1$}
\State{Initialize Worst-Case Perturbed Hamiltonian:\\
\mbox{$\tilde{H}^{(k)}(\delta_{1}) = H_0 + \sum\limits_{m=1}^M H_m f_{m}^{(k)} + \mathbf{d} \sum\limits_{m = 0}^{M} \alpha_{m}^{(k)} \hat{H}_{m} s_m^{(k)}(0)$}}
\State{Evaluate $\tilde{e}_\mu(\delta_1)$}
\While{$\epsilon - \tilde{e}_\mu(\delta_n)>0$}
\State{$n = n+1$, $\delta_n = n \mathbf{d}$}
\State{Compute $s_\mu^{(k)}(n-1)$}
\State{$\tilde{H}^{(k)}(\delta_n) = \tilde{H}^{(k)}(\delta_{n-1}) + \mathbf{d} \sum\limits_{m=0}^{M} \alpha_{m}^{(k)} \hat{H}_m s_m^{(k)}(n-1)$}
\State{Evaluate $\tilde{e}_\mu(\delta_n)$}
\EndWhile
\State{Set $\bar{n} = n - 1$, $\bar{
\delta} = \bar{n}\mathbf{d}$}
\end{algorithmic}
\end{algorithm}

\section{Case Study: Gate Optimization}\label{sec: example}

\begin{figure} \centering
\includegraphics[width = 1\columnwidth]{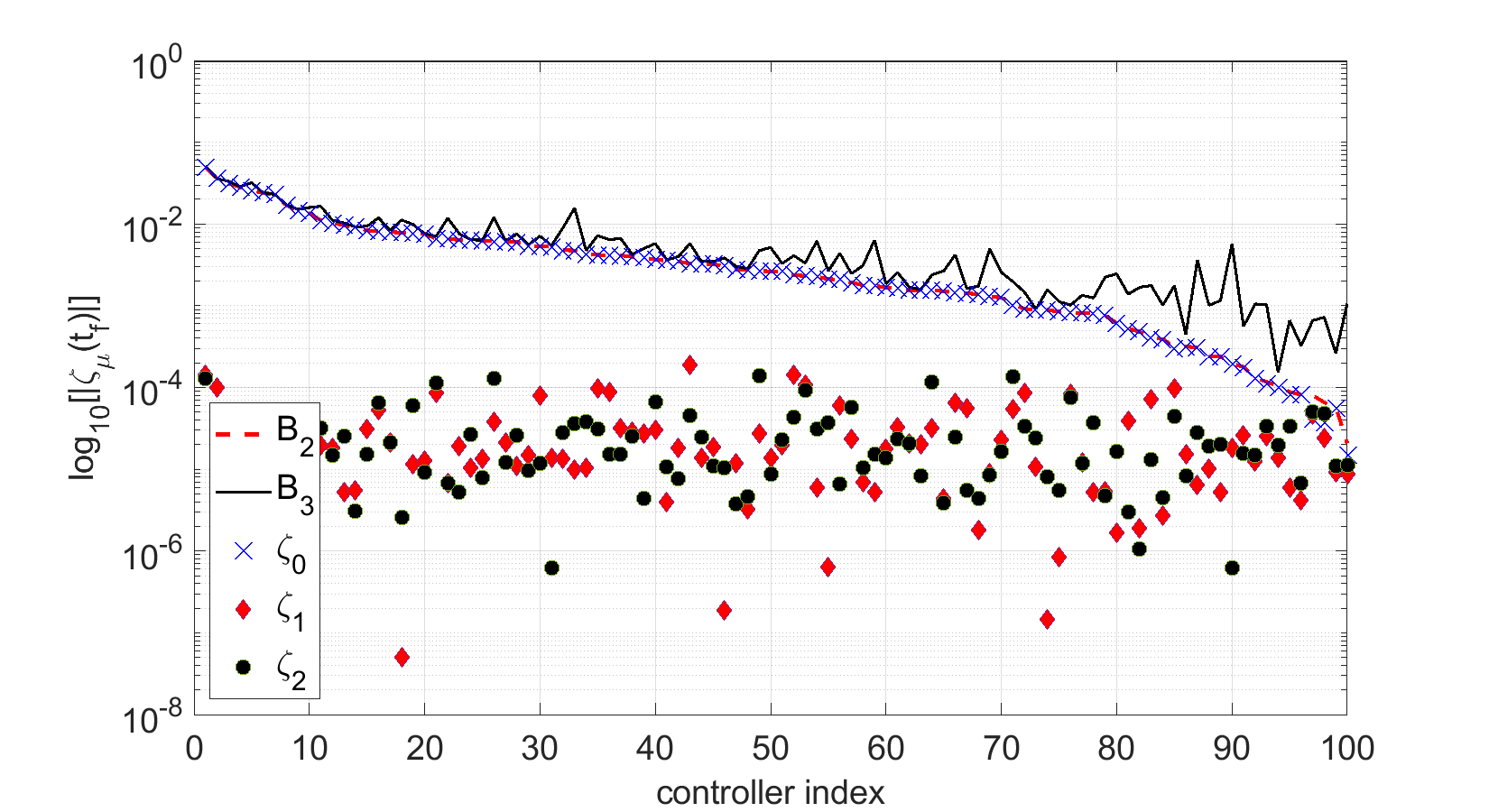}
\caption{Semilog plot of $ \left| \zeta_\mu(t_f) \right|$ versus controller index. The solid line indicates the upper bound $B_3$ for a variable uncertainty structure. The dashed line indicates the upper bound $B_2$ for a static structure. Markers of different shapes/colors indicate $\left| \zeta_\mu(t_f) \right|$ for the uncertainty structures $\hat{H}_\mu$ for $\mu \in \set{0,1,2}$. Note that $|\zeta_\mu(t_f)| \leq B_2 \leq B_3$ for all controllers.}  
\label{fig: tight_bounds}
\end{figure}

To illustrate these results, we consider dynamic controllers optimized for maximum gate fidelity in a three-spin chain with Heisenberg coupling~\cite{Floether_2012}. As opposed to full spin addressability, we consider the case of control applied only to the initial spin of the chain, a more challenging optimization problem.  The drift and interaction Hamiltonian matrices are:
\begin{align}
\begin{split} \label{eq: case_study_hamiltonians}
&H_0 = \frac{1}{2} \sum_{\ell = 1}^{2} \left( \sigma_{x}^{(\ell)}\sigma_{x}^{(\ell +1)} + \sigma_{y}^{(\ell)} \sigma_{y}^{(\ell+1)} + \sigma_{z}^{(\ell)}\sigma_{z}^{(\ell+1)} \right), \\
&H_{1} = \frac{1}{2} \sigma_{x}^{1}, \quad H_{2} = \frac{1}{2} \sigma_{y}^{1},
\end{split}
\end{align}
where $\sigma_{\set{x,y,z}}$ are the Pauli spin operators. Here $\sigma_{\set{x,y,x}}^{(\ell)}$ is the $3$-fold tensor product with $\sigma_{\set{x,y,z}}$ in the $\ell$th position and $I_2$ in the others. The target gate $U_f$ is a randomly generated unitary gate, and the initial gate is taken as the identify matrix $I_8$. The choice of a randomly-generated unitary gate as the target is used to increase the difficulty of the optimization. The gate operation time is $t_f = 15$ with $\kappa = 32$ time steps. 

\begin{figure}
\includegraphics[width=1\columnwidth]{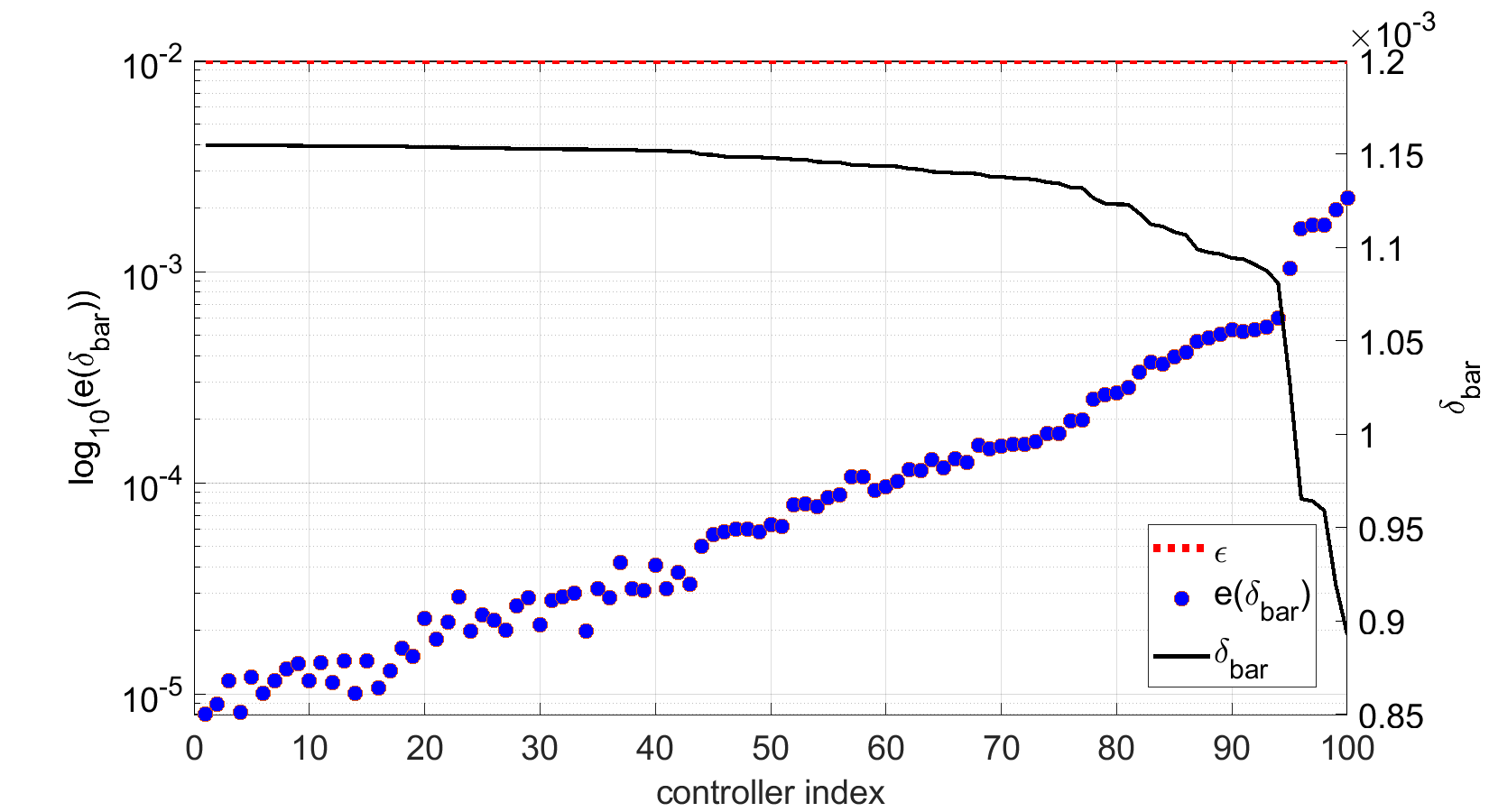}
\caption{
Semilog plot of $\tilde{e}_{0}(\bar{\delta})$ and $\bar{\delta}$ versus controller index. The solid line is $\bar{\delta}$ computed from Theorem~\ref{thm: theorem_3}. The dotted line is the performance threshold $\epsilon=0.01$. The markers indicate $\tilde{e}_0(\bar{\delta)}$, showing that for $\delta < \bar{\delta}$, $\tilde{e}_0(\delta)$ does not exceed $\epsilon$.}  
\label{fig: performance_loose}
\end{figure}

\begin{figure}
\includegraphics[width = 1\columnwidth]{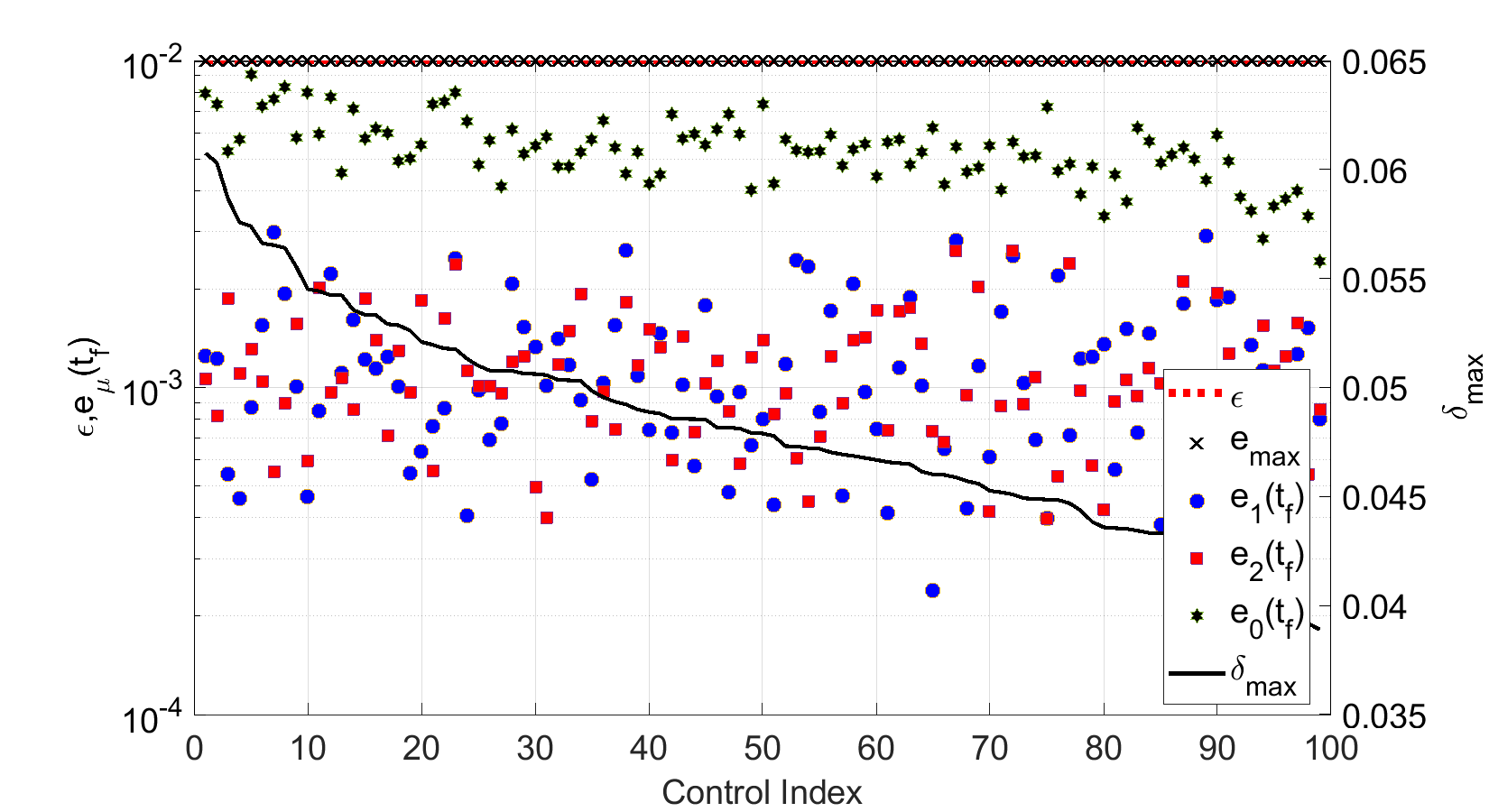}
\caption{
Semilog plot of $\tilde{e}_{\mu}(\bar{\delta})$ and $\bar{\delta}$ versus controller index. The solid line is $\bar{\delta}$ computed by iterating on $B_3$ from Theorem~\ref{thm: theorem_2}. The dotted line is the performance threshold $\epsilon = 0.01$. The x marker indicates $e(\bar{\delta})$ for the worst-case structure. The other markers indicate $e_\mu(\bar{\delta})$ for a perturbation $\bar{\delta}$ for $\mu \in \set{0,1,2}$. Note that no perturbation of size $\bar{\delta}$ in any principal direction $\hat{H}_\mu$ yields an error that exceeds $\epsilon$.} 
\label{fig: performance_tight}
\end{figure}

As discussed in Section~\ref{sec: sensitivity_bounds} we choose as an uncertainty structure basis $\{\hat{H}_\mu \}_{\mu=0}^2$. Each element  $\hat{H}_\mu$ is the normalized version of the matrices in~\eqref{eq: case_study_hamiltonians} as per Section~\ref{sec: uncertainty_model}.  We first examine the bound on all perturbations at $\delta = 0$. Figure~\ref{fig: tight_bounds} shows the tightness of the bounds $B_2$ and $B_3$ obtained from Theorems~\ref{thm: theorem_1} and~\ref{thm: theorem_2}. Note that the sensitivity resulting from the uncertainty structure $\hat{H}_0$ nearly matches $B_2$. As expected, the bound $B_3$ for a non-static uncertainty structure is slightly larger than $B_2$. For performance, we examine the closeness of the performance guarantees of Theorem~\ref{thm: theorem_3} to the actual perturbed error based on the predicted worst-case $\bar{\delta}$.  We set the value of $\epsilon$ at $0.01$, so that a gate fidelity of $99\%$ is the minimum acceptable performance.  As shown in Figure~\ref{fig: performance_loose}, perturbing the system in the direction $\hat{H}_0$ up to $\delta$ equating its strict inequality bound $\bar{\delta}$ given by Theorem~\ref{thm: theorem_3} does not violate the performance criterion $\tilde{e}_{0}(\bar{\delta}) < \epsilon$.  However, the plot also reveals the conservativeness of $\bar{\delta}$ calculated this way. Specifically, the predicted value of $\bar{\delta}$ results in a fidelity error $e_{0}(\bar{\delta})$ orders of magnitude below the minimum performance threshold for most controllers. In only six cases does $\tilde{e}_0(\bar{\delta})$ approach within an order of magnitude of the maximum allowable error $\epsilon$, providing some utility to the first order, differential-based performance bound.  Still, the result is indicative of the limitation of differential sensitivity techniques to guarantee performance for non-vanishing perturbations.  However, as noted in Section~\ref{sec: performance}, we can use Theorem~\ref{thm: theorem_2} to iterate a computational search for the minimum value of $\delta$ such that $\tilde{e}_\mu(\delta)$ exceeds $\epsilon$ for an arbitrary structure $\hat{H}_\mu$. Figure~\ref{fig: performance_tight} shows the results of iterating on $\delta$ to calculate this $\bar{\delta}$ as per Algorithm 1.  Values of $\bar{\delta}$ average two magnitudes greater than those predicted by Theorem~\ref{thm: theorem_3}.  As seen in the figure, the iterative procedure results in a $\bar{\delta}$ such that the sequence of perturbations pushes the fidelity error to the limit of the performance criteria $\epsilon$ indicated by $\tilde{e}_{max}$. Additionally, we note that in the perturbative regime around $\delta = 0$, perturbations structured as $\hat{H}_0$ show the greatest sensitivity. As seen in Figure~\ref{fig: performance_tight}, perturbations of this same structure lead to those values of $\tilde{e}_{\mu}(\delta)$ that most closely approach the $\epsilon$-threshold. This suggests that the differential sensitivity properties at $\delta=0$ are indicative of sensitivity to the error at non-vanishing values of $\delta$, at least for this controller set. 

\section{Conclusion}\label{sec: conclusion}
We have extended the calculation of the differential sensitivity in~\cite{Sean_sensitivity} to the case of optimal piecewise constant controls, shown that the differential sensitivity can be reliably bounded for small perturbations about the nominal operating point, and determined those uncertainty structures that induce the maximal sensitivity. A limitation of the approach is that differential techniques are applicable for small perturbations only. Performance guarantees for larger perturbations and comparison with statistical robustness measures requires further work.

\bibliographystyle{ieeetr}
\bibliography{bibliography}

\end{document}